\documentclass[aps, prl, reprint, showkeys, superscriptaddress]{revtex4-1}

\usepackage{blindtext}
\usepackage[english]{babel}
\usepackage{xcolor}
\usepackage{graphicx}
\usepackage[font=scriptsize, labelfont=bf, justification=justified]{caption}
\usepackage[noend]{algpseudocode}
\usepackage{algorithm}
\usepackage{amsfonts}
\usepackage{titlesec}

\usepackage{nameref}

\usepackage{amsmath}
\usepackage{color}
\usepackage{enumitem}
\usepackage{url}
\usepackage{amsthm}

\newtheorem{theorem}{Theorem}[section]
\newtheorem{lemma}[theorem]{Lemma}

\titleformat{\section}{\normalfont\fontsize{12}{12}\bfseries}{\thesection}{}{}
\titlespacing\section{0pt}{10pt plus 2pt minus 2pt}{6pt plus 2pt minus 2pt}

\titleformat{\subsection}{\normalfont\fontsize{11}{11}\bfseries}{\thesubsection}{}{}

\begin{document}
\title{Quantifying topological invariants of neuronal morphologies}
\author{Lida Kanari}
\affiliation{Blue Brain Project, EPFL}
\email{lida.kanari@epfl.ch}
\author{Pawe{\l} D{\l}otko}
\affiliation{DataShape, INRIA Saclay, Ile-de-France}
\author{Martina Scolamiero}
\affiliation{Laboratory for Topology and Neuroscience at the Brain Mind Institute, EPFL}
\author{Ran Levi}
\affiliation{Institute of Mathematics, University of Aberdeen}
\author{Julian Shillcock}
\affiliation{Blue Brain Project, EPFL}
\author{Kathryn Hess}
\affiliation{Laboratory for Topology and Neuroscience at the Brain Mind Institute, EPFL}
\author{Henry Markram}
\affiliation{Blue Brain Project, EPFL}

\date{\today}

\begin{abstract}
Nervous systems are characterized by neurons displaying a diversity of morphological shapes. Traditionally, different shapes have been qualitatively described based on visual inspection and quantitatively described based on morphometric parameters. Neither process provides a solid foundation for categorizing the various morphologies, a problem that is important in many fields. We propose a stable topological measure as a standardized descriptor for any tree-like morphology, which encodes its skeletal branching anatomy. More specifically it is a barcode of the branching tree as determined by a spherical filtration centered at the root or neuronal soma. This Topological Morphology Descriptor (TMD) allows for the discrimination of groups of random and neuronal trees at linear computational cost. 
\end{abstract}


\maketitle


The analysis of complex tree structures, such as neurons, branched polymers \cite{Alexandrowicz:1985}, viscous fingering \cite{Oded:2002} and fractal trees \cite{Mandelbrot:1983, Garst:1990}, is important for understanding many physical and biological processes. Yet an efficient method for quantitatively analysing the morphology of such trees has been difficult to find. Biological systems provide many examples of complex trees. The nervous system is one of the most complex biological systems known, whose fundamental units, neurons, are sophisticated information processing cells possessing highly ramified arborizations \cite{Jan:2010}. The structure and size of neuronal trees determines the input sources to a neuron and the range of target outputs and is thought to reflect their involvement in different computational tasks \cite{Cuntz:2007, Ferrante:2013, Silver:2010}. In order to understand brain functionality, it is fundamental to understand how neuronal shape determines neuronal function \cite{Honey:2010}. As a result, much effort has been devoted to grouping neurons into distinguishable morphological classes \cite{DeFelipe:2013, Markram:2004, Halavi:2012, PING:2008}, a categorization process that is important in many fields \cite{Lyons:1999}.

Neurons come in a variety of shapes with different branching patterns (e.g., frequency of branching, branching angles, branching length, overall extent of the branches, etc). Classifying these different morphologies has traditionally focused on visually distinguishing the shapes as observed under a microscope \cite{Masseroli:1993}. This method is inadequate as it is subject to large variations between experts studying the morphologies \cite{DeFelipe:2013} and is made even more difficult by the presence of an enormous variety of morphological types \cite{Markram:2004}. An objective method of discriminating between neuronal morphologies could advance progress in generating a “parts-list” of neurons in the nervous system. For this reason, experts now generate a digital version of the cell’s structure - a 3D reconstructed model of the neuron \cite{Dieter:2000} that can be studied computationally. The reconstructed morphology is encoded as a set of points in $\mathbb{R}^3$ along each branch and edges connecting pairs of points. The reconstruction forms a mathematical tree representing the neuron’s morphology.

\begin{figure}[!ht]
   \includegraphics[width=0.45\textwidth]{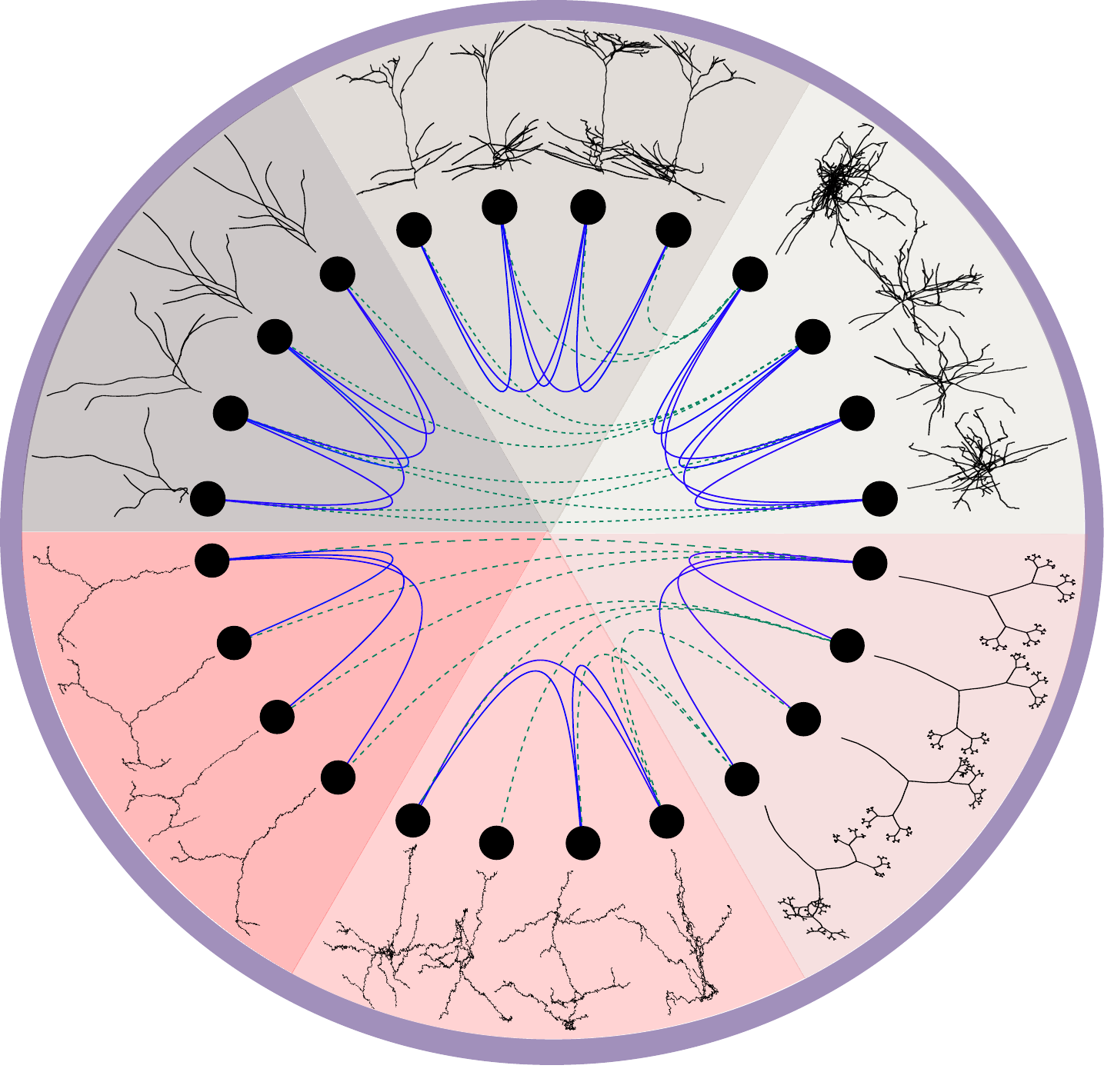}
 \caption{Illustration of the separation of similar tree structures into distinct groups, using topological analysis. The colored pie segments show six distinct tree types: three neuronal types (upper half) and three artificial ones (lower half). The thick blue lines show that our topological analysis can reliably separate similar-looking trees into groups. It is accurate both for artificially-generated trees and neuronal morphologies. The dashed green lines show that classification using an improper set of user-selected features cannot distinguish the correct groups.}
 \label{fig:grouping}
\end{figure}

In general, the properties of geometric trees, described by a set of points, have been rigorously studied in two extreme cases: in the limit of the full complexity of the trees \cite{Gunnar:2009}, where the entire set of points is used, and in the opposite limit of a feature space \cite{Gomez-Gil:2008, DeFelipe:2013}, where a (typically small) number of selected morphometrics (i.e., statistical features of the branching pattern) \cite{Blackman:2014, Wan:2015} are extracted from the digital morphology and form the input for standard classification algorithms. 

Topological data analysis (TDA) has been shown to reliably identify geometric objects that are built out of well-understood pieces, such as spheres, cylinders and tori, based on a point cloud sampled from the object \cite{Gunnar:2009}. It suffers, however, from the deficiency that reliable classification of complex geometric tree structures by standard TDA methods, such as Rips complexes \cite{Edelsbrunner:2008}, requires thousands of sampled points, which is expensive in terms of both computational complexity and memory requirements. 

Thus, feature extraction is the only currently feasible solution to establishing a more quantitative approach to morphologies\cite{Scorcioni:2008, Ling:2012, PING:2008}. While this approach has been efficiently used in image recognition \cite{Schurer:1994}, the wide diversity of neuronal shapes, even for cells identified as of the same type, has made it difficult to isolate an optimal set of features that reliably describe all neuronal shapes. Indeed, alternative sets of morphometrics result in different classifications \cite{DeFelipe:2013}, as illustrated in Fig ~\ref{fig:grouping}, because the statistical features commonly overlap even across markedly different morphological types. This in turn is because traditional feature extraction results in significant loss of information as the dimensionality of the data is reduced.





As a result, neither of these methods is suitable for the study of complex tree structures that describe physical or biological objects. 
The algorithm we propose here overcomes both of these limitations by constructing a Topological Morphology Descriptor (TMD) of a tree embedded in $\mathbb{R}^3$. A distance inspired by persistent homology \cite{Gunnar:2009} is defined to distinguish trees, because alternative methods for computing distances between trees, such as the \textit{edit distance} \cite{Bille_editingDistance}, the \textit{blastneuron distance} \cite{Wan:2015} and the \textit{functional distortion distance} \cite{Bauer:2014}, are in general computationally expensive (see SI,~\nameref{sec:dtree}). Our algorithm preserves the overall shape of the tree while reducing the computational cost by discarding the local fluctuations. It takes as input the set of key functional points in a tree \cite{Ferrante:2013}: the branch points, i.e., the nodes with more than one child, and the leaves, i.e., the nodes with no children and transforms them into a multi-set of intervals on the real line known as a \textit{persistence barcode} \cite{Gunnar:2009}, Fig ~\ref{fig:method}b. Each interval encodes the lifetime of each connected component in the underlying structure, delimited by the points at which a branch is first detected (birth) and when it connects to a larger subtree (death). This information can be equivalently represented in a \textit{persistence diagram}\cite{Gunnar:2009}, Fig ~\ref{fig:method}c in which the birth-death times are points in the real plane with the advantage that this structure simplifies the mathematical analysis of the data.



Our method is applicable to any tree-like structure and we demonstrate its generality by first applying it to a collection of artificial random trees, (see SI,~\nameref{sec:rtrees}), and to groups of neuronal trees (see SI,~\nameref{sec:data}). Our results show that the TMD of tree shapes is highly effective for reliably and efficiently distinguishing random and neuronal trees (Fig ~\ref{fig:grouping}).


\begin{figure}[!ht]
 \centering
   \includegraphics[width=0.49\textwidth]{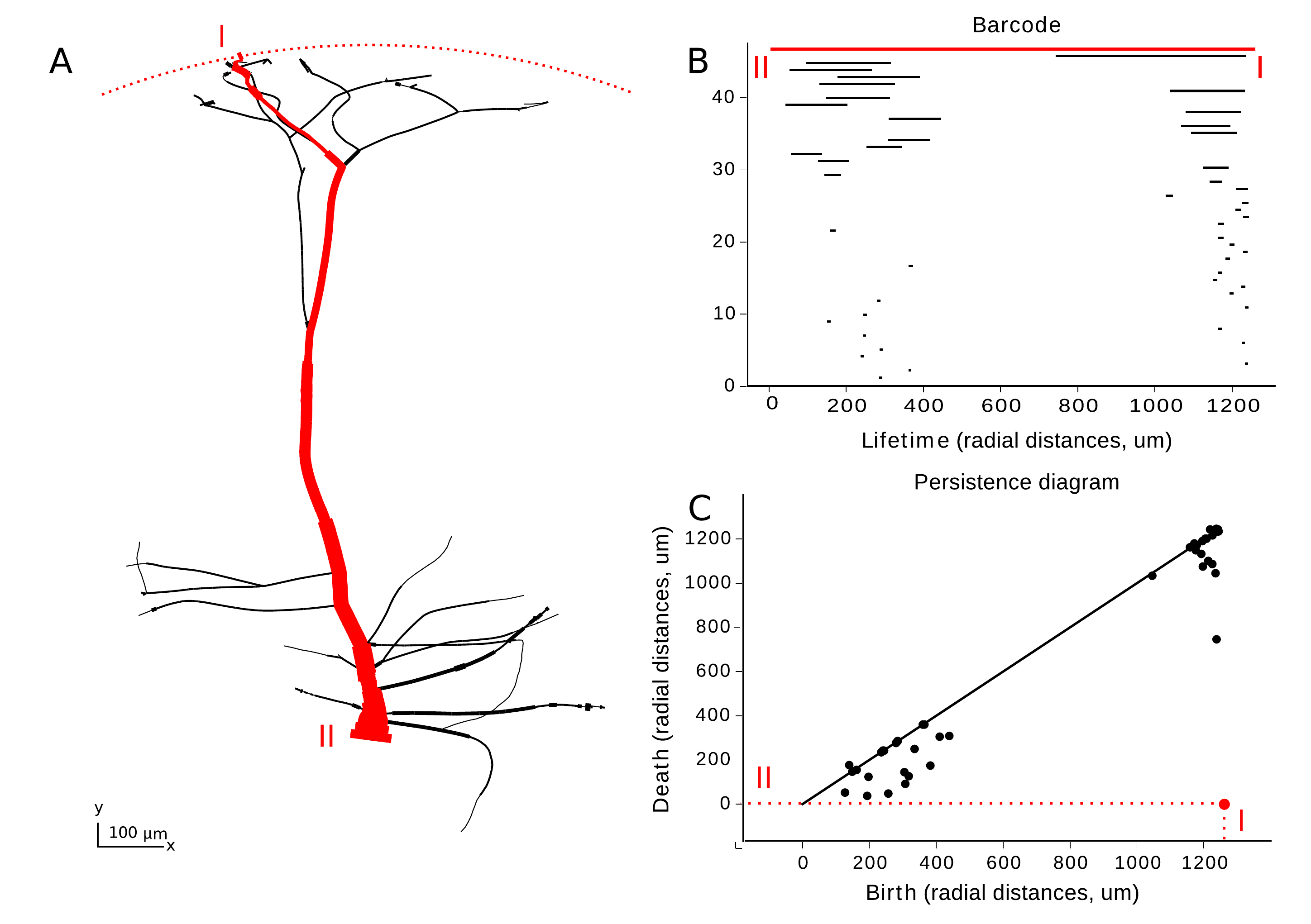}
 \caption{Application of topological analysis to a neuronal tree (A) showing the largest persistent component (red). The persistence barcode (B) represents each component as a horizontal line whose endpoints mark its birth and death in units that depend on the choice of the function $f$ used for the ordering of the nodes of the tree. In our case, it is radial distance of the nodes from the soma, so the units are microns. The largest component is again shown in red together with its birth (I) and death (II). This barcode can be equivalently represented as points in a persistence diagram (C) where the birth (I) and death (II) of a component are the X and Y coordinates of a point respectively (in red). The diagonal line is a guide to the eye and marks points whose birth and death are simultaneous. Note that the persistence barcode (B) has an arbitrary ordering of the components along the Y axis, which prevents the calculation of an average barcode. The corresponding persistence diagram (C), however, has physical coordinates on both axes, allowing an average persistence diagram to be defined.}
 \label{fig:method}
\end{figure}


\section{Methods} \label{methods}

The operation $\mathcal{F}$ to extract the barcode $D(T)$ of an embedded tree $T$ is described in Algorithm 1. Let $T$ be a rooted, and therefore oriented, tree \cite{Knuth:1998}, embedded in $\mathbb{R}^3$. Note that the operation described here is in fact generalizable to trees embedded in any metric space. We denote by $N := B \cup L$ the set of nodes of $T$, which is the union of the set of branch points $B$ and the set of leaves $L$. In the case of a neuron, the root $R$ is the point representing the soma. Each node $n \in N$ has references to its parent, i.e., the first node on the path towards the root and to its children. Nodes with the same parent are called siblings. 

Let $f$ be a real-valued function defined on the set of nodes $N$. For the purpose of this study $f$ is the radial distance from the root $R$. Another possible choice would be the path distance from $R$. Based on the function $f$, we construct a real-valued function $v$, defined on the set of nodes $N$ as the largest value of the function $f$ on the leaves of the subtree starting at $n$. An ordering of the siblings can be defined based on $v$: if $n_1, n_2 \in N$, are siblings and $v(n_1)<v(n_2)$, then $n_1$ is younger than $n_2$.


The algorithm is initialized by setting the value of $v(l), l \in L$ equal to the value of $f(l)$. Following the path of each leaf $l \in L$ towards the root $R$, all but the oldest (with respect to $v$) siblings are killed at each branch point. If siblings are of the same age it is equivalent to terminate either one of them. For each killed component one interval (birth-death) is added to the persistence barcode (Fig ~\ref{fig:method}). The oldest sibling survives until it merges with a larger component at a subsequent branch point. This operation is applied iteratively to all the nodes until the root is reached. At this point only one component, the oldest one, is still alive. 

\begin{algorithm}[H]
\begin{algorithmic}[1]
\Require{$T$ with $R$, $B$, $L$, $f: R \cup L \rightarrow \mathbb{R}$ and $v$}
\State{$D(T):$ empty list to contain pairs of real numbers}
\State{$A \gets L$} \Comment{$A:$ set of active nodes}
\For{every $l \in L$}
\State $v(l) = f(l)$
\EndFor
\While{$R \not \in A$}
\For{$l$ in $A$}
\State{$p:$ parent of $l$}
\State{$C:$ children of $p$}
\If{$\forall n \in C, n \in A$}
\State{$c_m:$ randomly choose one of $\{c \mid v(c) = \max_{c'}(v(c'))$ for $c' \in C\}$}
\State{Add $p$ to $A$}
\For{$c_i$ in $C$}
\State{Remove $c_i$ from $A$}
\If{$c_i  \neq c_m$}
\State{Add ($v(c_i)$, $f(p)$) to $D(T)$}
\EndIf
\EndFor
\State{$v(p) \gets v(c_m)$} 
\EndIf
\EndFor
\EndWhile
\State{Add ($v(R)$, $f(R)$) to $D(T)$}
\State{Return $D(T)$}
\end{algorithmic}
\caption{$\mathcal{F}:$ extracting persistence barcode from tree}\label{alg:topaz}
\label{alg:topaz}
\end{algorithm}

When all the branches are outgoing, i.e., the radial distance of the origin is smaller than the radial distance of the terminal point of the branch, the operation $\mathcal{F}$ is equivalent to a filtration of concentric spheres of decreasing radii, centered at $R$ (Fig ~\ref{fig:method}). In this case, the birth time of a component is the supremum of the radii of the spheres that do not contain the entire component. The death time is the infimum of the radii of the spheres that contain the branch point at which the component merges with a longer one. 

The computational complexity of this operation is linear in the number of nodes. Note that the $if$ statement in line 9 of the algorithm is critical for the linear complexity. The number of currently active children is saved at each parent node to avoid quadratic complexity.

This process results in a set of intervals on the real line, each of which represents the lifetime of one component of the tree. The operation $\mathcal{F}$ that associates a persistence barcode $D(T)$ to a tree $T$ is invariant under rotations and translations, as long as the function $f$ is also. For the purposes of this paper, $f$ is the radial distance from $R$ and as such it is invariant under rotations about the root and rigid translations of the tree in $\mathbb{R}^3$.




The most common topological metric that is used for the comparison of persistence diagrams is the \textit{bottleneck distance} \cite{Edelsbrunner:2008}, denoted $d_B$. The bottleneck distance is the infimum of the infinity norm $L_{\infty}$ (maximum distance) between matching points over all possible matchings of two persistence diagrams \cite{Edelsbrunner:2008}. 

We prove that the operation $\mathcal{F}: T \mapsto D(T)$ is stable with respect to standard topological distances, such as the Bottleneck distance (see SI,~\nameref{sec:stab1}). For $\epsilon$ modifications in the tree $T$ with respect to Hausdorff distance (see SI,~\nameref{sec:dspace}), the persistence diagram $D(T)$ is not modified more than $\mathcal{O}(L\epsilon)$, as long as function $f$ is $L$-Lipschitz. Note that this requirement is always satisfied (with $L=1$) by the radial distance. As a result, the method is robust with respect to small perturbations in the positions of the nodes and the addition/ deletion of small branches.

\begin{figure}[!ht]
 \centering
   \includegraphics[width=0.49\textwidth]{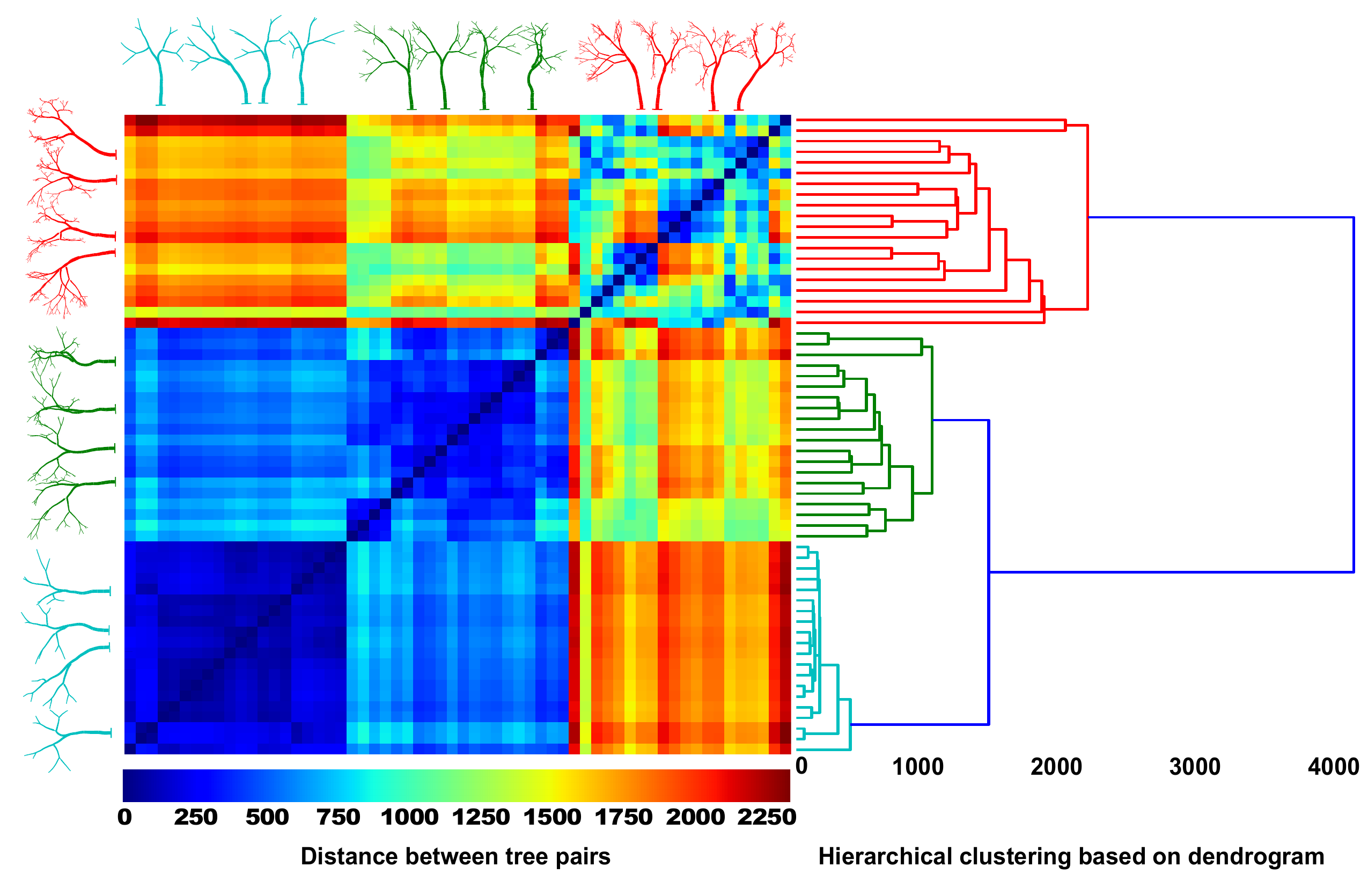}
 \caption{Application of topological analysis to a set of artificial trees generated using a well-defined stochastic process. Three sets of trees are shown in which each set differs from the others only in the tree depth. Each group contains 20 individuals of which only 4 are shown for clarity. Each individual is generated using the same tree parameters but a different random number seed. The TMD separates each group accurately using the distance metric between barcodes of the trees described in the main text. This is seen in the separation of colors in the colormap, and is supported by the clear distinction into three groups resulting from a simple hierarchical clustering algorithm [34].}
 \label{fig:randomtrees}
\end{figure}


However, none of the standard topological distances between persistence diagrams is appropriate for the grouping of neuronal trees. The bottleneck distance as well as distances based on the bottleneck, such as the \textit{persistence distortion distance} \cite{metric_distortion} (see SI,~\nameref{sec:dtree}) cannot distinguish diagrams that differ on the short components, which are nevertheless important for the distinction of neuronal morphologies. 

\begin{figure*}[!ht]
 \centering
   \includegraphics[width=0.95\textwidth]{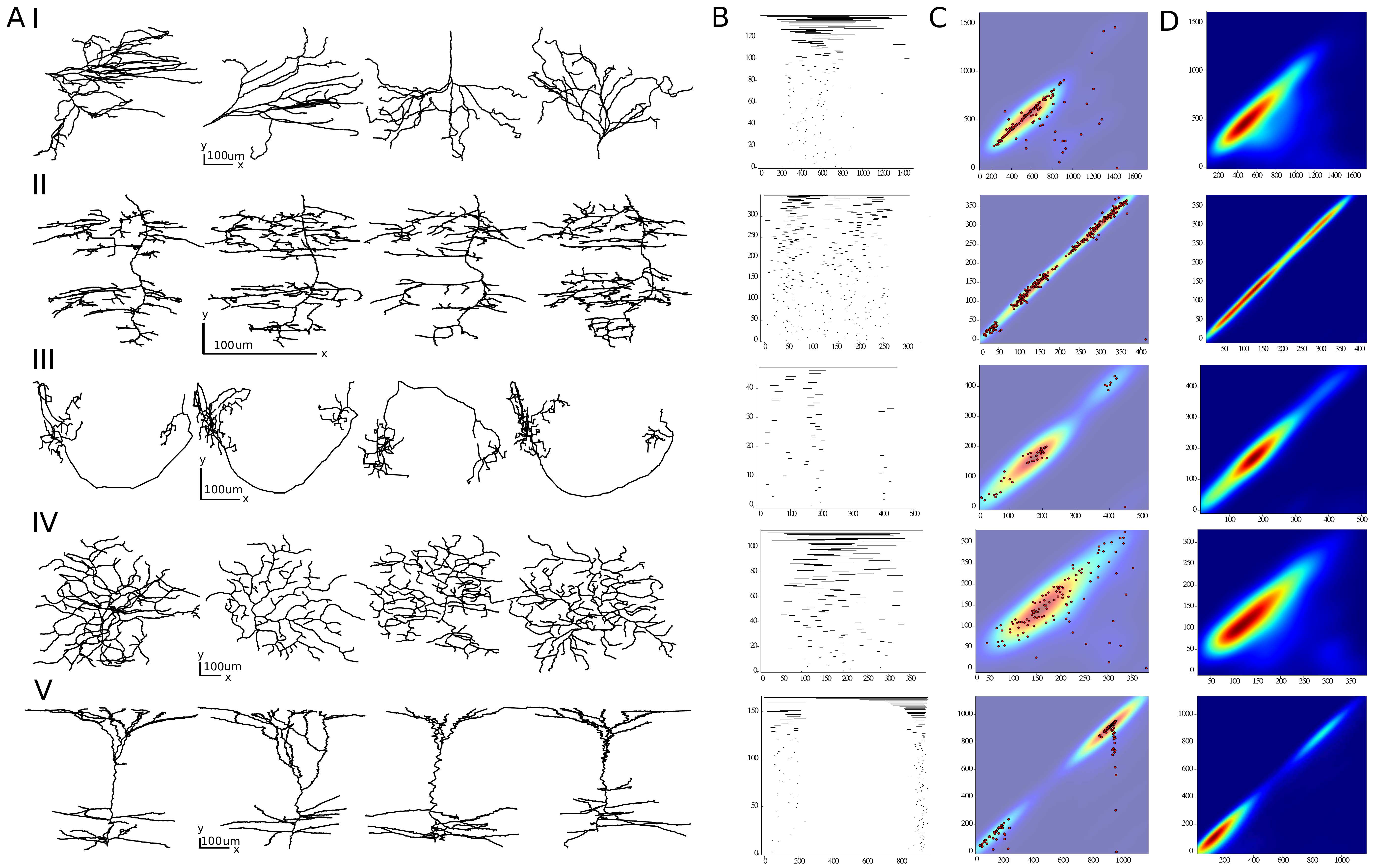}
 \caption{Topological comparison of neurons from different animal species. Each row corresponds to a species: (I) cat, (II) dragonfly, (III) drosophila, (IV) mouse and (IV) rat. Trees from several exemplar cells for each species are shown in the first column (A). Representative persistence barcodes for the cells in A are shown in the second column (B). The structural differences of the trees are clearly evident in these barcodes. II, III and V have clusters of short components, clearly distinct from the largest component, while I and IV have bars of a quasi-continuous distribution of decreasing lengths. Also, barcodes III, and V show empty regions between dense regions of bars, indicating the existence of two clusters in the morphologies, while barcodes I and IV are dense overall. The persistent image for each representative barcode in B and its superimposed persistence diagram are shown in the third column (C). By combining the persistence diagrams in (C) for several tress we can define an average persistent image (D) in order to study the qualitative differences between sets of morphologies. The trees in the first row (cat) are more tightly grouped than those in the second row (dragonfly), and two clusters are visible in the dragonfly trees. Considering rows 1 and 4, the extension of the elliptical peak perpendicular to the diagonal line reflects the variance in the length-scale mentioned earlier for a single cell's barcode. Note that the trees, barcodes, and persistent images are not shown to the same scale: see the scale bar in each case. }
 \label{fig:species}
\end{figure*}

We therefore define an alternative distance $d_{Bar}$ in the space of the barcodes that we use to classify branching morphologies. For each barcode we generate a density profile as follows: $\forall x \in \mathbb{R}$ the value of the histogram is the number of intervals that contain $x$, i.e., the number of components alive at that point. The distance between two barcodes $D(T_1)$ and $D(T_2)$ is defined as the sum of the differences between the density profiles of the barcodes. This distance is not stable with respect to Hausdorff distance, but it is the only distance we are aware of that succeeds in capturing the differences between distinct neuronal persistence barcodes.


Finally, the persistence diagram can be converted into a \textit{persistent image}, as introduced in \cite{Chepushtanova:2015}. This method is based on the discretization of a sum of Gaussian kernels \cite{Scott:2008}, centered at the points of the persistence diagram. This discretization can generate a matrix of pixel values, encoding the persistence diagram in a vector, called the persistent image. Machine learning tools, such as decision trees and support vector machines can then be applied to this vectorization for the classification of the persistence diagrams.


\section{Results}
We demonstrate the discriminative power of the TMD by applying it to three examples of increasing complexity. First, we applied it to artificial random trees that provide a well-defined test case to explore the method's performance because they are generated by a stochastic algorithm (described in ~\nameref{sec:rtrees}) and therefore have properties that can be precisely modified. Next, we analyzed two datasets of more biological relevance: neurons from different species, downloaded from \cite{Ascoli:2007}, and distinct types of trees obtained from several morphological types of rat cortical pyramidal cells (see SI,~\nameref{sec:data}). This last example is interesting because, although there is biological support for their separation into distinct groups, no consistent objective classification has previously been performed. The persistence diagrams of trees from all these example datasets can be classified correctly with high accuracy, indicating good performance of the algorithm across non-trivial examples. 


Mathematical random trees are defined by a set of parameters that constrain their shape: the tree depth $T_d$, the branch length $B_l$, the bifurcation angle $B_a$ and the degree of randomness $D_r$ (see SI,~\nameref{sec:genrtrees}). We defined a control group as a set of trees generated with predefined parameters ($T_d=5$, $B_l=10$, $B_a=\pi/4$, $D_r=10\%$) but independent random seeds. Each parameter was varied individually to generate groups of trees that differed from the control group in only one property. For all trees we extracted the persistence barcode using Algorithm 1. A comparison of the distances $d_{Bar}$ between each tree's barcode and the barcodes of the trees in every groups constitutes one trial. The trial is successful if the minimum barcode distance occurs for a tree and the group generated with the same parameters.

The TMD method accurately separates groups of random trees that have different \emph{tree depth} as seen in Fig~\ref{fig:randomtrees}, with an accuracy of $99\% \pm 2\%$ (see SI,~\nameref{sec:rtresults}). The distance matrix in Fig ~\ref{fig:randomtrees} shows three distinct groups and a simple hierarchical clustering algorithm \cite{Ward:1963} correctly assigns each tree to its group. The TMD’s effectiveness is demonstrated by its accuracy in assigning individual trees to their originating group. The TMD distinguishes random trees generated by varying the parameters $T_d$, $B_a$, $B_l$, $D_r$ with an accuracy of $99 \%$, $94 \%$, $99 \%$ and $77 \%$ respectively. Further results are given in the ~\nameref{sec:rtresults}. 


Next, the algorithm is used to quantify differences between neuronal morphologies. Neurons that serve distinct functional purposes exhibit unique branching patterns. The morphologies of neuronal trees from different species, as seen in Fig ~\ref{fig:species}, can be classified into different groups. For the purposes of this study, we used cat, drosophila, dragonfly, mouse and rat neurons. The qualitative differences between the neuronal tree types are evident from the individual geometrical tree shapes (Fig ~\ref{fig:species}A) as well as the extracted barcodes (Fig ~\ref{fig:species}B).

\begin{figure}[!ht]
 \centering
   \includegraphics[width=0.49\textwidth]{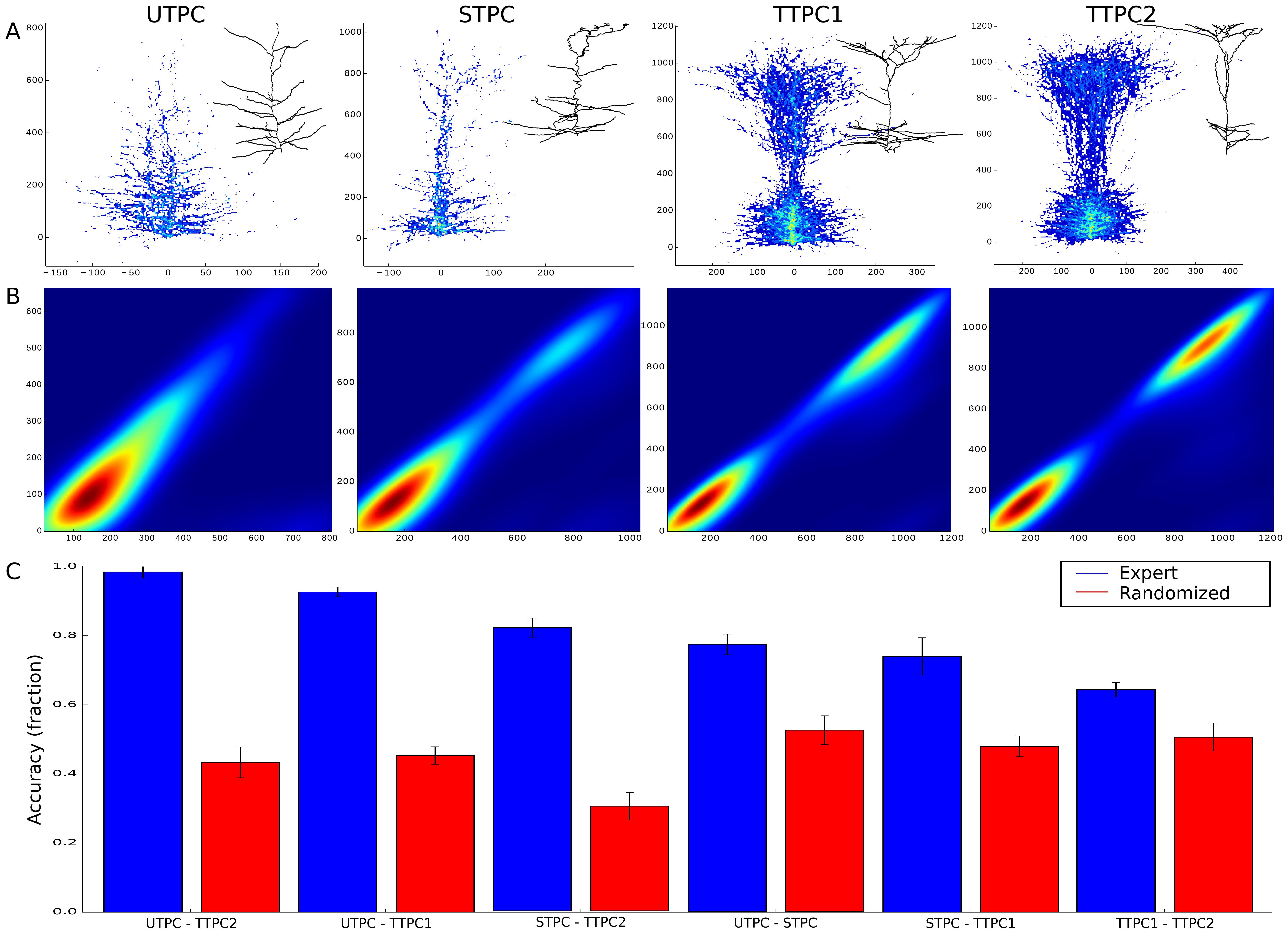}
 \caption{Comparison of the topological analysis of apical dendrite trees extracted from several types of rat neuron. Four cell types are shown in (A): UTPC, STPC, TTPC1, TTPC2 (left to right). The morphological differences between these cell types are subtle, but the persistent images (B) clearly reveal them, particularly the presence of two clusters in the TTPC1 and TTPC2 cell types. From these persistent images we train a decision tree classifier on the expert-assigned groups of cells. Then we perform binary classification on trees taken from each pair of groups (C). The blue bars show that the performance of the classifier using the expert-assigned groups in order of decreasing accuracy (98\%, 93\%, 82\%, 77\%, 74\%, 64\%). The red bars are obtained from a random assignment of trees to groups and are significantly worse than the expert assignment. This shows that the persistent images objectively support the expert's classification.}
 \label{fig:pcs}
\end{figure}


From the superposition of persistence diagrams for a set of trees, we constructed the persistent images of the group (Fig ~\ref{fig:species}D). Note that the persistence barcode has an arbitrary ordering of the tree’s components along the Y-axis, which prevents the calculation of an average barcode. The corresponding persistence diagram, and the associated persistent image, have physical coordinates on both axes, allowing an average persistence diagram to be defined. This representation is useful for quantifying the differences between tree types. It is trivial to identify in Fig ~\ref{fig:species}D areas of high branching density at different radial distances from the soma by the position of the clusters. Since density is usually correlated with connection probability, we can identify the anatomical parts of the trees that are important for the functionality and the connectivity patterns of different cell types. 

Finally, we applied the TMD to a case for which it is difficult for a non-expert to distinguish the morphologies. While pyramidal cell (PC) morphologies (Fig~\ref{fig:pcs}A) of the rat appear superficially similar, the persistent images (\ref{fig:pcs}B) reveal fundamental morphological differences between the four neuronal types, related to the existence and the shape of the apical tuft. The apical tuft of PCs is known to play a key role in the integration of neuronal inputs through their synapses in higher cortical layers, and is therefore a key indicator for the functional role of the cell. Therefore, the TMD can successfully discriminate between these functionally distinct neuronal types. 




\section{Discussion}

The morphological diversity of neurons supports the complex information processing capabilities of the nervous system. A major challenge in neuroscience has therefore been to objectively classify the neurons by shape. Long-standing techniques, such as feature extraction, have so far been unable to reliably distinguish cell types, without expert manual interaction, in a reasonable computational time. We have introduced here a Topological Morphology Descriptor derived from principles of persistent topology that retains enough topological information about trees to be able to distinguish them in linear computational time. Its ability to efficiently separate neurons into distinct groups with high accuracy provides objective support for expert-assigned neuronal classes, and may lead to a more comprehensive understanding of the relation between neuronal morphology and function.

This technique can be applied to any rooted tree equipped with a Lipschitz function defined on its nodes. Further biological examples include botanic trees and the roots of plants. The method is not restricted to trees in $\mathbb{R}^3$, however, but can be generalized to any subset $T$ of a metric space $M$, with a base-point $R$. A persistence barcode can then be extracted using a filtration by concentric spheres in $M$ centered at $R$, enabling us to efficiently study the shape of complex multidimensional objects. Another generalization is to use a different function $f$, such as the path distance from the root, which should serve to reveal shape characteristics that are independent of the radial distance and thus not captured by the current approach.

While the static neuronal structures presented in this paper are biologically interesting themselves, our method can also be used to track the morphological evolution of trees. The topological study of the growth of an embedded tree could be addressed through Multidimensional Persistence \cite{Carlsson:2009}, as the spherical filtration identifying relevant topological features of the tree could be enriched with a second filtration representing temporal evolution. This application could be useful in agriculture to study growing roots \cite{Wang:2009} and trees, and would also extend the method to neurons in the developing brain.

\section{Acknowledgments}

LK and JCS were supported by Blue Brain Project. P.D. and R.L. were supported in part by the Blue Brain Project and by the start-up grant of KH. Partial support for P.D. has been provided by the Advanced Grant of the European Research Council GUDHI (Geometric Understanding in Higher Dimensions). MS was supported by the SNF NCCR "Synapsy".

\section{Author contributions}
LK and KH conceived the study. LK developed the topological algorithm with the contribution of PD. PD, KH contributed topological ideas used in the analysis. HM suggested the biological datasets analyzed. RL, MS and JCS gave conceptual advice. LK, PD, MS, JCS, KH and HM wrote the main paper. All authors discussed the results and commented on the manuscript at all stages.

The authors declare no competing financial interests.

\clearpage

\onecolumngrid

\section*{\centering{Supplementary information}}

\section{Demonstration of the algorithm} \label{sec:alg}

\begin{figure}[!ht]
 \centering
   \includegraphics[width=0.8\textwidth]{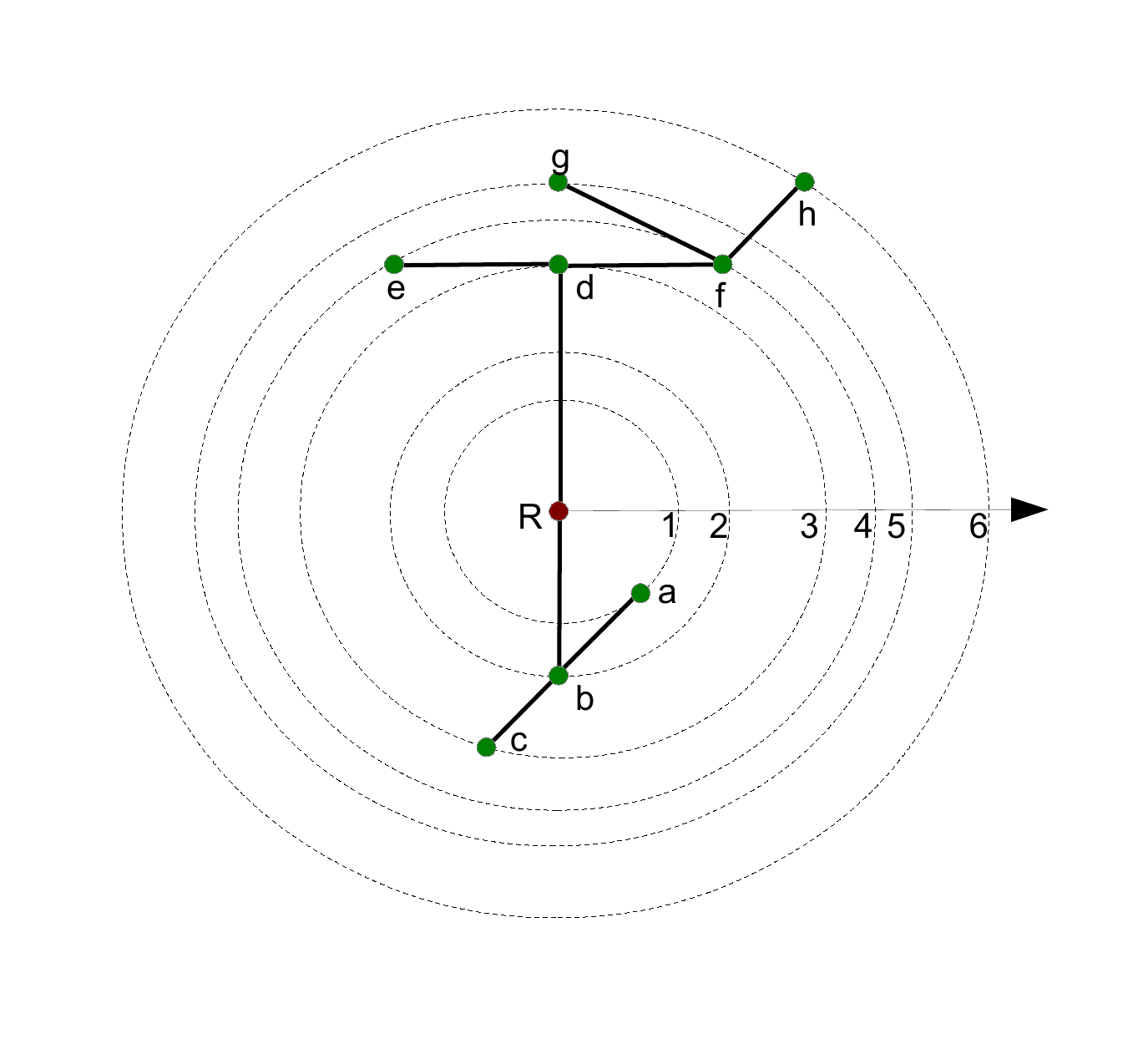}
 \caption{Illustration of the Algorithm~1.}
 \label{fig:algorithm}
\end{figure}

The idea of Algorithm~1 is presented in Figure~\ref{fig:algorithm}. The root, denoted by $R$, is shown in red, while the other nodes of the tree are labelled $a-h$. During the initialization of the algorithm, the leaves $(a,c,e,g,h)$ are inserted into the list of alive nodes, $A$. The algorithm then iterates over the members of $A$, as follows. 

If the node $a$ is the first element of the list $A$, then the first step of the algorithm begins by finding its parent, $b$. Since both $a$ and $c$ are in $A$, the algorithm then computes $c_m = c$ and therefore $v(b) = 3$. The interval $[1,2]$ is added to the list of intervals $D(T)$ representing the lifetime of the node $a$. Nodes $a$ and $c$ are removed from $A$, and $b$ is added to $A$. 
If the next vertex in the list $A$ is $e$, the algorithm begins the next step by finding its parent, $d$. This node is not processed further at this stage, however, since its sibling, $f$, is not in $A$. The next node to be processed is therefore $g$. Both children $g$ and $h$ of its parent $f$ are in $A$. The oldest component is $c_m = h$ and therefore $v(f) =6$. The interval $[5,4]$ is added in $D(T)$ to represent the lifetime of node $g$. The node $f$ is added to $A$, and both $g$ and $h$ are removed from $A$. The list of alive components then consists of $b,e,$ and $f$. The node $b$ cannot be processed since its sibling $d$ is not in $A$. The next node to be processed is therefore $e$, whose parent $d$ has all of its children in $A$. In this case, $c_m = f$, and $v(d) = v(f) = 6$. The interval $[4,3]$ is added to $D(T)$, while nodes $e$ and $f$ are removed from $A$, and $d$ is added to $A$. The next node in $A$ is $b$, whose parent is $R$. Since both children of $R$ are now in $A$, we compute $c_m = d$, $v(R) = v(d) = 6$,  and add the interval $[3,2]$ to $D(T)$. The  root $R$ is added to $A$, and $b$ and $d$ are removed from the list. Since now only the root $R$ is alive, the while loop in the algorithm terminates. The last step adds the interval $[6,0]$ to $D(T)$.

\section{Definition of Distances} 

\subsection{Distances between metric spaces} \label{sec:dspace}

\label{sec:hausdorff_distance}

In order to analyze the stability of Algorithm~1, we need to define a notion of distance between trees embedded in a metric space, as well as of distance between persistence diagrams. 

We will use the standard Hausdorff distance, defined as follows, to measure the distance between trees. Given $T_1, T_2$, trees in embedded in a metric space with  metric $d$, the Hausdorff distance between $T_1$ and $T_2$ is given by the formula:

\[ d_H(T_1,T_2) = max( sup_{x \in T_1} inf_{y \in T_2} d(x,y) ) , sup_{x \in T_2} inf_{y \in T_1} d(x,y) ). \]

\subsection{Distances between persistence diagrams} \label{sec:dph}

\label{sec:pers_distances}

Below we recall various representations of persistence diagrams and the possible notions of distance between them. We also provide a reference to software that computes the distances considered if such exist. All of the metrics summarized below can be applied directly to the output of Algorithm~1.

The most classical distances used in topological data analysis are the bottleneck and Wasserstein distances. Given two
 persistence diagrams $D_1$ and $D_2$, we may assume without loss of generality that they both contain the points from the diagonal with infinite multiplicity. We construct a matching (i.e., a bijection) $\phi : D_1 \rightarrow D_2$ and define two numbers 

 $$B_{\phi} = \sup_{x \in D_1}{ d(x,\phi(x)) }$$

 and 

 $$W^p_{\phi} = ( \sum_{x \in D_1} d(x,\phi(x))^p )^{\frac{1}{p}}, $$

\noindent where $d$ is the standard euclidean distance in $\mathbb{R}^2$. Note that $B_{\phi}$ is simply the longest distance that $\phi$ shifts a point in $D_1$, while $(W^p_{\phi})^p$ is a sum of $p$-th powers of lengths of the line segments joining $x$ and $\phi(x)$, for all $x$. The infimum of $B_{\phi}$ over all possible matchings is the bottleneck distance between $D_1$ and $D_2$. The infimum of $W^p_{\phi}$ over all possible matchings is the $p-$Wasserstein distance between $D_1$ and $D_2$. Given two persistence diagrams $D_1$ and $D_2$, their bottleneck distance will be denoted by $d_B(D_1,D_2)$ and their $p-$Wasserstein distance by $W_p(D_1,D_2)$.  One implementation of these distances can be found in~\cite{dionysus} and a faster approximation in~\cite{Kerber:2016}. Please follow the link in the paper for the implementation.
 
A persistence diagram can also be represented by a persistence landscape, i.e., a piecewise linear function $f : \mathbb{R} \times \mathbb{N} \rightarrow \mathbb{R}$. Given two persistence landscapes, we can compute the distance  between them in $L^p$ space ~\cite{Bubenik:2015}. The implementation is described in~\cite{Dlotko:2015}.

One can also represent persistence diagrams by persistence images, which were introduced in \cite{Chepushtanova:2015}. The idea is to place a Gaussian kernel at every point of the diagram and to discretize the distribution obtained in a pixel-based image. It is then straightforward to compute a distance between two persistence images, using common image-recognition techniques. This representation is used in the classification of morphological types of neurons in the experimental section of this paper. We are not aware of a publicly available implementation of this approach. An implementation is provided with the software of this paper. 

\subsection{Distances between trees}\label{sec:dtree}

A classic metric to compare trees, the \textit{edit distance} \cite{Bille_editingDistance}, is based on the transformation of one tree $T_1$ into another $T_2$ by a sequence of operations (deletion and insertion of vertices), each of which has a non-negative cost. The edit distance \cite{Bille_editingDistance} between $T_1$ and $T_2$ is defined to be the supremum of the total cost of all possible transformations from $T_1$ to $T_2$. However, the edit distance is not relevant to our problem, since it does not involve geometric information about the tree structure and  is known to be NP-complete \cite{Feragen:2012}.  

A metric that is useful for distinguishing neuronal trees was proposed in \cite{Wan:2015}. A set of morphological features is extracted from the trees, and the initial estimation of the distance between them is defined by the distance between the extracted features. An alignment algorithm is then applied to pairs of trees in order to identify local similarities. The local alignment requires the comparison of all pairs of branches, making the computation very expensive. This method is designed for the efficient matching of trees with highly similar structures, but the high variability within the groups of rat cortical neurons does not allow similar trees to be grouped together by local alignment, since local structures are often altered. 

Another important class of distances between trees are the distances between merge trees introduced in \cite{Beketayev:2014} and \cite{Morozov:2013}, which are applied to merge trees of sublevel sets of functions. For a function $f : X \rightarrow \mathbb{R}$, where $X$ is a metric space, the sublevel set at level $a \in \mathbb{R}$ is $\{x \in X | f(x) \leq a\}$. The differences captured by merge trees are considerably more subtle than the differences captured by the persistent homology of the function's sublevel sets. The authors of \cite{Beketayev:2014} and \cite{Morozov:2013} provide examples of pairs of simple merge trees $T$ and $T'$ that have the same persistence diagrams, but that are a nonzero distance apart. It is clear that in this particular case, Algorithm~1 would provide the persistent homology of the sublevel sets of the function. Therefore, by rescaling $T$ and $T'$, the difference between the distances used in those papers and the distances used in the current paper can get arbitrarily large. 

Another relevant metric, the \textit{persistence distortion distance} \cite{metric_distortion} between two graphs $G_1$ and $G_2$, is expressed in terms of the shortest paths from every $v_1 \in G_1$ and from every $v_2 \in G_2$. Given this shortest path metric, the persistence distortion is defined as the minimal bottleneck distance between the zero-dimensional persistent homology modules of the superlevel sets of the distance functions $P(G_1,v_1)$ and $P(G_2,v_2)$.
The persistence diagrams obtained in the process of computing the persistence distortion distance are conceptually very close to the diagrams we get from Algorithm~1. In our case, we obtain a significant computational advantage from working with rooted trees, since there is always a unique path between every pair of vertices. Moreover the logical choice of initial vertices $v_1$ and $v_2$ from which to compute shortest paths is to take the root of the trees considered. In this case, the persistence diagram arising when computing the persistence distortion distance is the one we would get from Algorithm~1 when the function $v$ is the path distance from the root. 
The computational cost of the distortion distance is considerable in the general case, but linear in our case. However, since it is based on the bottleneck distance, it suffers from that metric's limitations, i.e., the shortest components, which are important for the neuronal morphologies, are not taken into account. The code to compute persistence distortion distance is available here~\cite{persistence_distortion_soft}.


\textit{Functional distortion distance} was first introduced in \cite{Bauer:2014} as a distance to compare Reeb graphs. Since this distance, more generally, compares graphs equipped with a function defined on edges, we will consider it for neural tree comparison. 

Let $T$ be a tree, viewed as a CW-complex, with a real-valued function $g\colon T \rightarrow \mathbb{R}$ that is monotone on edges. Following \cite{Bauer:2014} we give $T$ a metric space structure. For any two points  $u,v$ in $T$ (not necessarily vertices), let $\lambda$ be the unique path between $u$ and $v$. The height of $\lambda$ is defined as $\text{height}(\lambda):=\text{max}_{x \in \lambda} g(x) - \text{min}_{x \in \lambda} g(x)$. In particular, if $e=(v_0,v_1)$ is an edge in $T$, since $g$ is monotone on edges, $\text{height}(e)=|g(v_0)-g(v_1)|$.
The distance between two points $u$ and $v$ in $T$ is defined as $d_f(u,v):=\text{height}(\lambda)$.
The function $d_f$ defines a metric on $T$ whenever there is no path of constant function value between two points $u \ne v$ in $G$. Otherwise, $d_f$ is a pseudo-metric.

Let $T$ and $T^{\prime}$ be two trees equipped with real-valued functions $f \colon T \to \mathbb{R}$ and $f^{\prime}\colon T^{\prime}\to \mathbb{R}$. For every pair of continuous functions $\phi: T\to T^{\prime}$ and $\psi \colon T^{\prime} \to T$, let
$$G(\phi,\psi):=\{(x,\phi(x))| x \in T\} \cup \{(\psi(y),y)| y \in T^{\prime}\}$$
and
$$\mathcal{D}(\phi,\psi):=\text{sup}_{(x,y),(\tilde{x},\tilde{y})\in G(\phi,\psi)}\frac{1}{2} | d_f (x,\tilde{x})- d_{f^{\prime} }(y,\tilde{y})|.$$

The functional distortion distance between $T$ and $T^{\prime}$ is defined by:
$$d_{FD}(T,T^{\prime})= \text{inf}_{\phi,\psi}\text{max}\{ \mathcal{D}(\phi,\psi),||f-f^{\prime}\circ \phi||_{\infty},||f \circ \psi -f^{\prime}||_{\infty}\}.$$

In Section 3.2 we will use the properties of functional distortion distance to prove that Algorithm 1 is robust with respect to specific types of perturbation in the input data.

\section{Proof of Stability}

\subsection{Bottleneck and Wasserstein stability of the barcode construction} \label{sec:stab1}

We prove now that the operation $\mathcal{F}$ to transform a tree $T$ into a persistence diagram $D(T)$ is robust under $\epsilon$ small perturbations. In the space of metric trees, we use the Hausdorff distance defined in Section~\pageref{sec:hausdorff_distance}. Since our input consists only of leaves and branch points, and their respective connectivity, when perturbing the tree, we perturb only those points. In the space of persistence diagrams we consider the bottleneck and Wasserstein distances defined in Section~\nameref{sec:pers_distances}.

As a result, the method is robust with respect to noisy data. Given a tree $T$, with a real-valued function $f$ defined on the set of nodes $N$, we construct a real-valued function $v$, defined on the set of nodes $N$ as the largest value of the function $f$ on the leaves of the subtree starting at $n$ (see Methods). In order to prove the robustness of the method, we require the function $v$ to be Lipschitz with a constant $L$. For the purpose of this study $f$ is the radial distance from the root $R$ and as a result $v$ is $1-$Lipschitz.

The proof of stability is carried out in three steps. First we prove stability in the case when the branch point's position can modified by $\epsilon$. We then consider the removal or addition of short branches. At the end, we unite these results to prove the main theorem. 

\begin{figure}[!h]
\begin{minipage}[t]{0.45\linewidth}
    \centering
    \includegraphics[trim=40 40 20 20, width=1\textwidth]{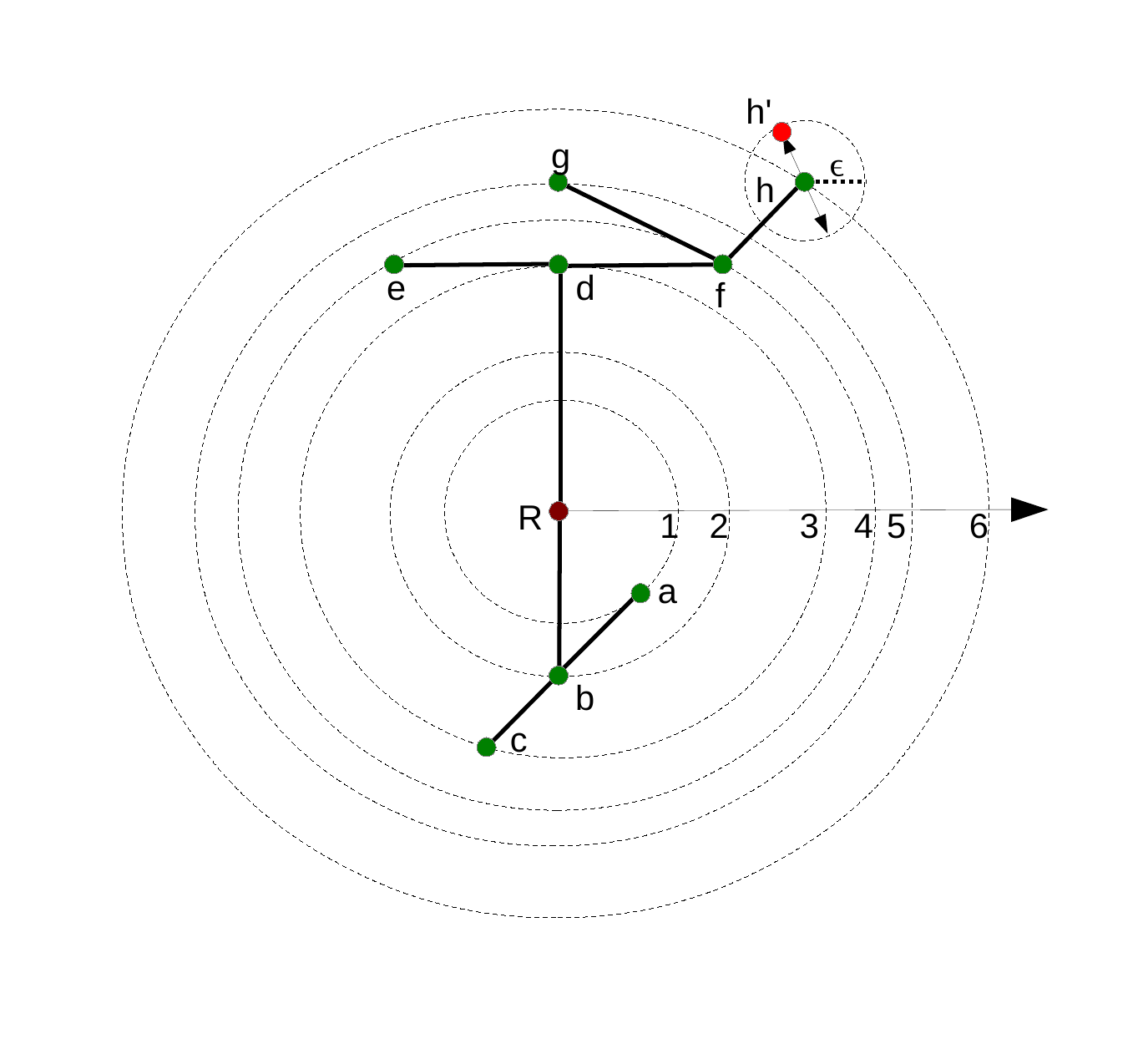}
    \caption{Perturbation of node position: h' has been moved $\epsilon$ from h.}
    \label{fig:pert1}
\end{minipage}
\hspace{0.2cm}
\begin{minipage}[t]{0.45\linewidth} 
    \centering
    \includegraphics[trim=30 20 20 20, width=1\textwidth]{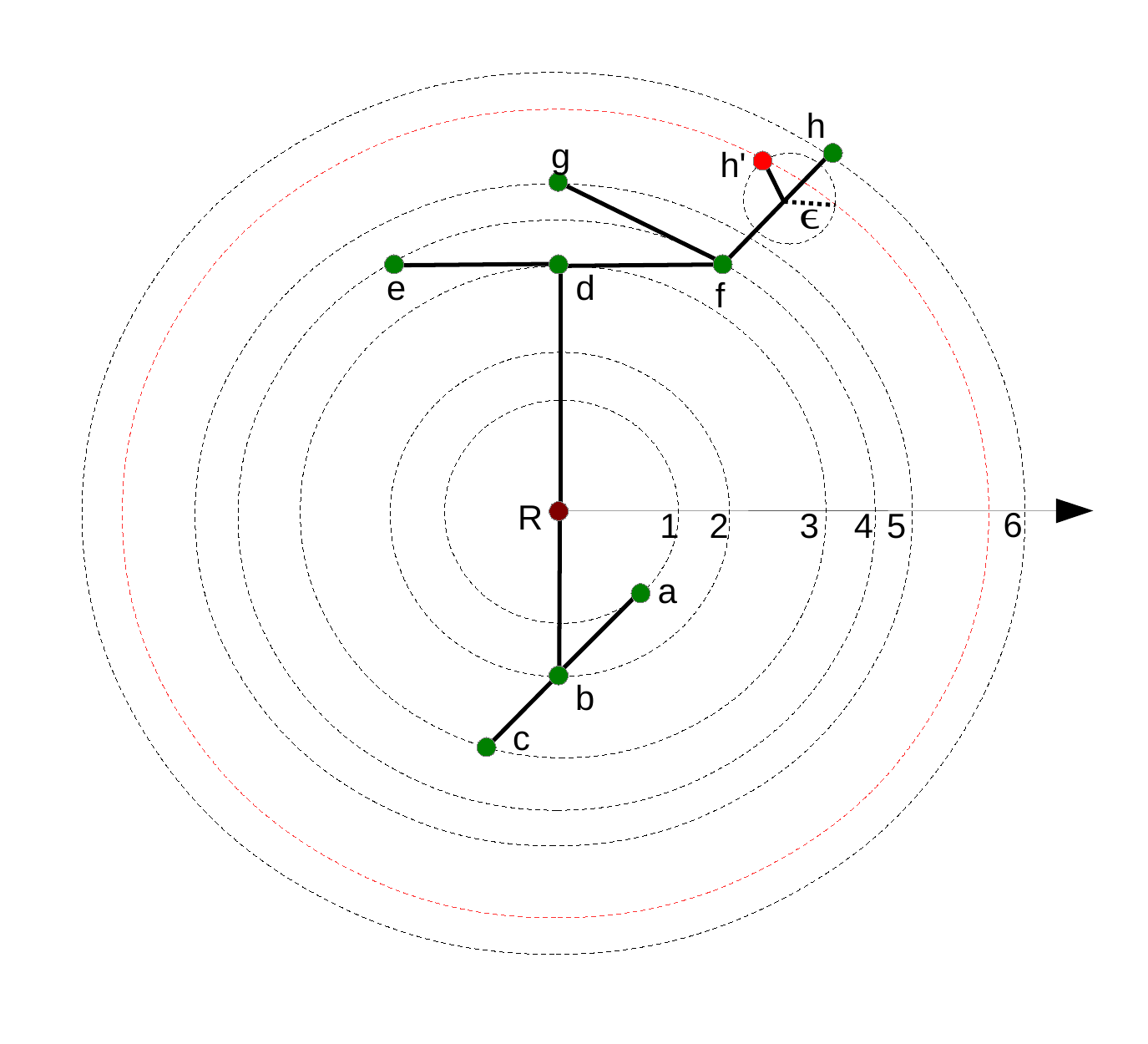}
    \caption{Addition of a small branch: h' has been added $\epsilon$ close to h.}
	\label{fig:pert2}
\end{minipage}
\end{figure}

We denote the perturbed tree as $T'$, with the same function $v$ defined on its nodes. We consider the following types of perturbations, for which $d_H(T,T') \leq \epsilon$.

\begin{enumerate}[label=(\alph*)]
\item  Perturbation of branch points' positions 

The location of a branch point is modified by at most $\epsilon$, as illustrated in Fig ~\ref{fig:pert1} for node $h$. The branching structure of the tree is preserved. In this case, the value of function $v$ of the node $h$ of tree $T$ will be modified by at most $L\epsilon$ in the node $h'$ of tree $T'$, since $v$ is L-Lipschitz. Since the connectivity has not been modified, the transformation of $T'$ to the persistence diagram $D(T')$, using Algorithm [1], will result in the same number of points. The difference between $D(T)$ and $D(T')$ will be in the position of the point(s) that depend on the modification of the node $h'$ and on $v(h')$. Assuming that the correspondence between the nodes of $T$ and $T'$ is known, all the points in $D(T')$ that have not been modified are matched with the corresponding points in $D(T)$. Since $|v(h)-v(h')| < L\epsilon$ the distance between the corresponding point(s) that have been modified in the persistence diagram will be at most $L\epsilon$. In this case the bottleneck distance between $D(T)$ and $D(T')$ is bounded by $L\epsilon$ and the Wasserstein distance in bounded by $cL\epsilon$, where $c$ is the number of points influenced by the modification of the node $h'$. 

We can generalize this concept to a perturbation of any number of nodes of the tree $T$. The maximum difference on the persistence diagrams can occur if both nodes of a component are moved by $\epsilon$ in opposite directions. In this case the distance between the two points is at most $2L\epsilon$. Therefore, the bottleneck distance between $D(T)$ and $D(T')$ is bounded by $2L\epsilon$. The $p-$th Wasserstein distance between $D(T)$ and $D(T')$ is bounded by $2n^{\frac{1}{p}}L\epsilon$, where $n$ is the number of leaves in $T$, which is equal to the number of points in $D(T)$.

Note that if the function $v$ depends on the location of the root, for example in our case $v$ is based on the radial distance $v$ from $R$, we need to double the Lipschitz constant in the analysis above. When a node of the tree is moved by at most $\epsilon$, the value of the function $v$ from the root may change by $2L\epsilon$ (respectively $2\epsilon$ in our case) since the position of both the root and any other node of a tree will be perturbed by at most $\epsilon$. Therefore the distance between the root and the node will change by no more than $2\epsilon$ and the value of the function $v$ on the node will change by no more than $2L\epsilon$.

%
%
\item Small branches are added to the tree. 

In this case, branches of at most length $\epsilon$ are added to the tree. The locations of the other nodes of the tree $T'$ remain fixed. In the simplest case, a single branch $h'$ is added to the component that contains node $h$, as presented in Fig ~\ref{fig:pert2}. This new branch will result in one more point in the corresponding persistence diagram. We create a matching between the points of $D(T)$ and $D(T')$. Assuming that the correspondence of the nodes of $T$ and $T'$ is known, all the points in $D(T')$ that have not been modified are matched with the corresponding points in $D(T)$. The new point will represent the termination of either the $h'$ or the $h$ component, depending on which node the value of function $v$ is smaller. In both cases, the new point will be matched with the diagonal at a distance less than $L\epsilon$ and therefore the bottleneck and the Wasserstein distances between $D(T)$ and $D(T')$ are bounded by $L\epsilon$.

We can generalize this concept to the addition of $m$ $\epsilon$ small branches to the tree $T$. In this case, $m$ new points will be added to the persistence diagram of $T'$. The rest of the points in $D(T')$ will be matched with the points of the corresponding branches in $D(T)$. Using the previous argument, the bottleneck distance between $T$ and $T'$ is bounded by $L\epsilon$. The $p-$th Wasserstein distance between $T$ and $T'$ is bounded by $2m^{\frac{1}{p}}L\epsilon$, where $m$ is the number of new at most $\epsilon$ long branches in $T'$ that are not in $T$.

\item Small branches are removed from the tree. 

Another perturbation is to remove small branches of length at most $\epsilon$ from the tree. This case is the inverse of the previous case. As a result, stability in this case is implied by the stability of the previous case.
\item Generalization of previous results. 

Finally, we discuss the most general case, where we can have perturbation of branch points' positions as well as addition and deletion of small branches. From the previous cases, using the triangle inequality, it is easy to show that the bottleneck distance between $T$ and $T'$ is bounded by $3L\epsilon$ and the $p-$th Wasserstein distance between $T$ and $T'$ is bounded by $3(m+n)^{\frac{1}{p}}L\epsilon$, where $m$ is the number of branches of length $\epsilon$ that are added or removed, and $n$ is the number of branch points in $T$.

\end{enumerate} 

\begin{figure}[!ht]
 \centering
   \includegraphics[width=0.95\textwidth]{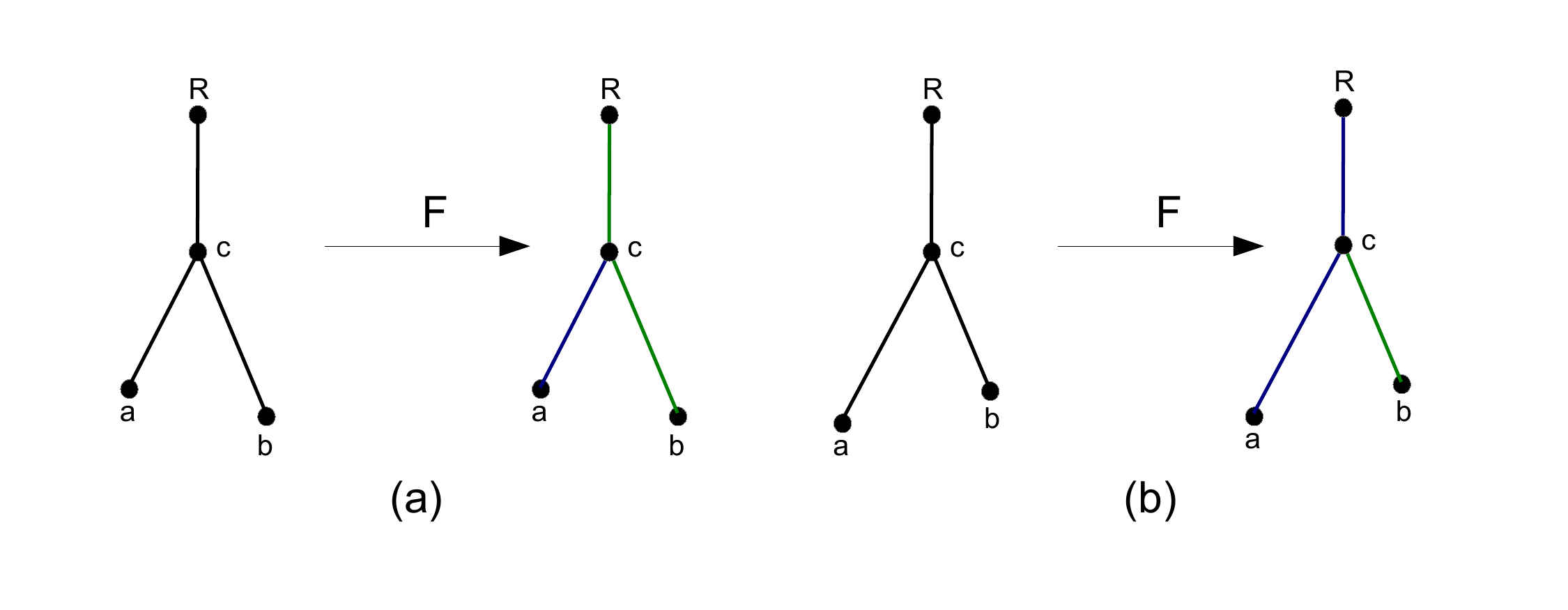}
 \caption{An illustration of a region of a tree corresponding to persistence intervals computed by Algorithm~1.}
 \label{fig:interval-decomposition-of-a-tree}
\end{figure}

Note that the persistence intervals introduced by Algorithm~1 also give rise to a decomposition of a tree into a distinct collection of branches. As an example, let us consider a tree in Figure~\ref{fig:interval-decomposition-of-a-tree}.a. The point $R$ denotes the root, $a$ and $b$ are leaves while $c$ is the only branch point in the tree. Function $v$ depends on $f$ which is the radial distance from the root $R$. In Figure~\ref{fig:interval-decomposition-of-a-tree}.a we assume that $b$ is slightly further away from $R$ than $a$. In this case, when Algorithm~1 processes the node $c$, the branch originating at $a$ will terminate at $c$ and the branch originating at $b$ will continue. This gives rise to a branch decomposition illustrated on the right corner. However, if the leaves $a$ and $b$ are slightly perturbed, as in Figure~\ref{fig:interval-decomposition-of-a-tree}.b, the resulting branch decomposition in the lower right corner will be very different from the previous case even though the persistence intervals will be very close for those two cases. This example illustrates that persistence diagrams are stable, but the corresponding branch decompositions are not. Since Algorithm~1 extracts the persistence diagram of a tree, the instability of the branch decomposition do not concern us for the current paper. 

\subsection{Stability proof based on functional distortion distance} \label{sec:stab2}

In this section we will use properties of the functional distortion distance $d_{FD}$ to prove that sampling errors in the input neural tree will not result in a significant difference in the associated persistence diagram. The space of persistence diagrams is equipped with the bottleneck distance $d_B$.

Recall that a neural tree is represented by a tree $T$ in $\mathbb{R}^3$ with a selected vertex $R$ called the root and a function $f \colon T \to \mathbb {R}$ associating to each vertex in $T$ its radial distance to the root. In Algorithm $1$ the tree structure is represented by the set of branch points, the set of leaves and their pairwise connectivity. In this section instead, in order to use the results from \cite{Bauer:2014}, the tree is viewed as a regular one-dimensional CW-complex. Intuitively we have glued an homeomorphic image of the interval $[0,1]$ between any two adjacent vertices in the data structure for Algorithm $1$.
In accordance with \cite{Bauer:2014} we denote by $Dg_0(T_f)$ the $0$-th persistence diagram of the sublevel set filtration of  $f \colon T \to \mathbb {R}$. The persistence diagram of the superset filtration of $f$ is equivalent to $Dg_0(T_{-f})$.


Let $T$ and $T^{\prime}$ be trees, equipped with functions $f \colon T \to \mathbb{R}$ and $f^{\prime} \colon T^{\prime} \to \mathbb{R}$.
Theorem $4.2$ of \cite{Bauer:2014} states that:  \[d_B(Dg_0(T_{f}),Dg_0(T^{\prime}_{f^{\prime}})) \leq d_{FD}(T,T^{\prime}).\]
 
The persistence diagram $D(T)$ obtained from Algorithm $1$ is equivalent to considering the disjoint union of the persistence diagrams $Dg_0(T_{f})$ and $Dg_0(T_{-f})$. In particular $Dg_0(T_{f})$ 
corresponds to points above the diagonal in $D(T)$ and identifies leaves that are closer to the root than their corresponding branch point. Points in $Dg_0(T_{-f})$ are instead under the diagonal in $D(f)$ and represent leaves that are farther from the root than their corresponding branch point. The diagonal of a persistence diagram is considered with infinite multiplicity, which implies that a matching between $D(T)$ and $D(T^{\prime})$ corresponds to a  pair: a matching of $Dg_0(T_{f})$ with $Dg_0(T^{\prime}_{f^{\prime}})$ and a matching of $Dg_0(T_{-f})$ with $Dg_0(T^{\prime}_{-f^{\prime}})$. It is therefore sufficient to consider the optimal matchings on the upper diagonal and lower diagonal part of the diagram given by Theorem $4.2$ to conclude that:
 \[d_B(D(T),D(T^{\prime})) \leq d_{FD}(T,T^{\prime}).\]
 
We will now use the property above to prove that the bottleneck distance between the persistence diagrams of trees in $\mathbb{R}^3$ is stable with respect to adding branches of length at most $\epsilon$ and perturbing the position of a branching point by $\epsilon$. 

\begin{lemma}
Let $T$ be a tree in $\mathbb{R}^3$ and $T^{\prime}$ a tree obtained by adding an edge of length at most $\epsilon$ to $T$, i.e. $T^{\prime}=T \cup_v e$.
If $f$ and $f^{\prime}$ are functions on $T$ and $T^{\prime}$ that associate to each point of a rooted tree its euclidean distance to the root, then $d_{FD}(T, T^{\prime} ) < \epsilon$.
\end{lemma}

\begin{proof}
Let $\phi: T \hookrightarrow T^{\prime}$ be the natural inclusion of $T$ in $T^{\prime}$, and let $\psi: T^{\prime} \to T$ be defined by:
\[ \psi(x) =
  \begin{cases}
    x      & \quad \text{if } x \notin e \\
    v & \quad \text{if } x \in e.\\
  \end{cases}
\]
By definition $f = f \circ \phi$. By the triangle inequality, $||f \circ \psi -f ||_{\infty} \leq \epsilon$. Furthermore since the length of the edge $e$ is at most $\epsilon$, $\mathcal{D}(\phi,\psi) \leq \epsilon/2.$ Therefore $d_{FD}(T,T^{\prime}) \leq \epsilon$.
\end{proof}

The proof above can immediately be generalised to the case in which finitely many edges, each of length at most $\epsilon$, are added to the tree. 

\begin{lemma}
Let $h \colon \mathbb{R}^3 \to \mathbb{R}^3$ be a homeomorphism such that the euclidean distance between $x$ and $h(x)$ is at most $\epsilon$ for any $x$ in $\mathbb{R}^3$.
If $f$ and $f^{\prime}$ are functions on $T$ and $h(T)$ that associate to each point of a rooted tree its euclidean distance to the root, then $d_{FD}(T, h(T)) < \epsilon$.
\end{lemma}

\begin{proof}
Since the euclidean distance between $x$ and $h(x)$ is at most $\epsilon$, the triangle inequality implies that $||f- f^{\prime} \circ h ||_{\infty}\leq \epsilon $. Let $g \colon \mathbb{R}^3 \to \mathbb{R}^3$ be the inverse to $h$. As before since the distance between $x$ and $g(x)$ is at most $\epsilon$, the triangle inequality implies that  $|| f \circ g - f^{\prime} ||_{\infty}\leq \epsilon$. 
By the definition of $h$ and $g$, it is also true that  $\mathcal{D}(h,g)< \epsilon$. It follows that $d_{FD}\big(T,h(T)\big) \leq \epsilon.$
\end{proof}

We have now proved that adding branches of length at most $\epsilon$ and moving the points of the tree $T$ in $\mathbb{R}^3$ by $\epsilon$ results in a tree which is $\epsilon$-close to $T$ in functional distortion distance. In this section we also showed that the bottleneck distance between persistence diagrams of trees is stable with respect to functional distortion distance.
Combining these two facts  guarantees that under the mentioned $\epsilon$-perturbations, the bottleneck distance between the persistence diagrams is at most $\epsilon$.

\section{Random Trees}  \label{sec:rtrees}
\subsection{Generation of random trees} \label{sec:genrtrees}

\begin{figure}[!ht]
 \centering
   \includegraphics[width=0.49\textwidth]{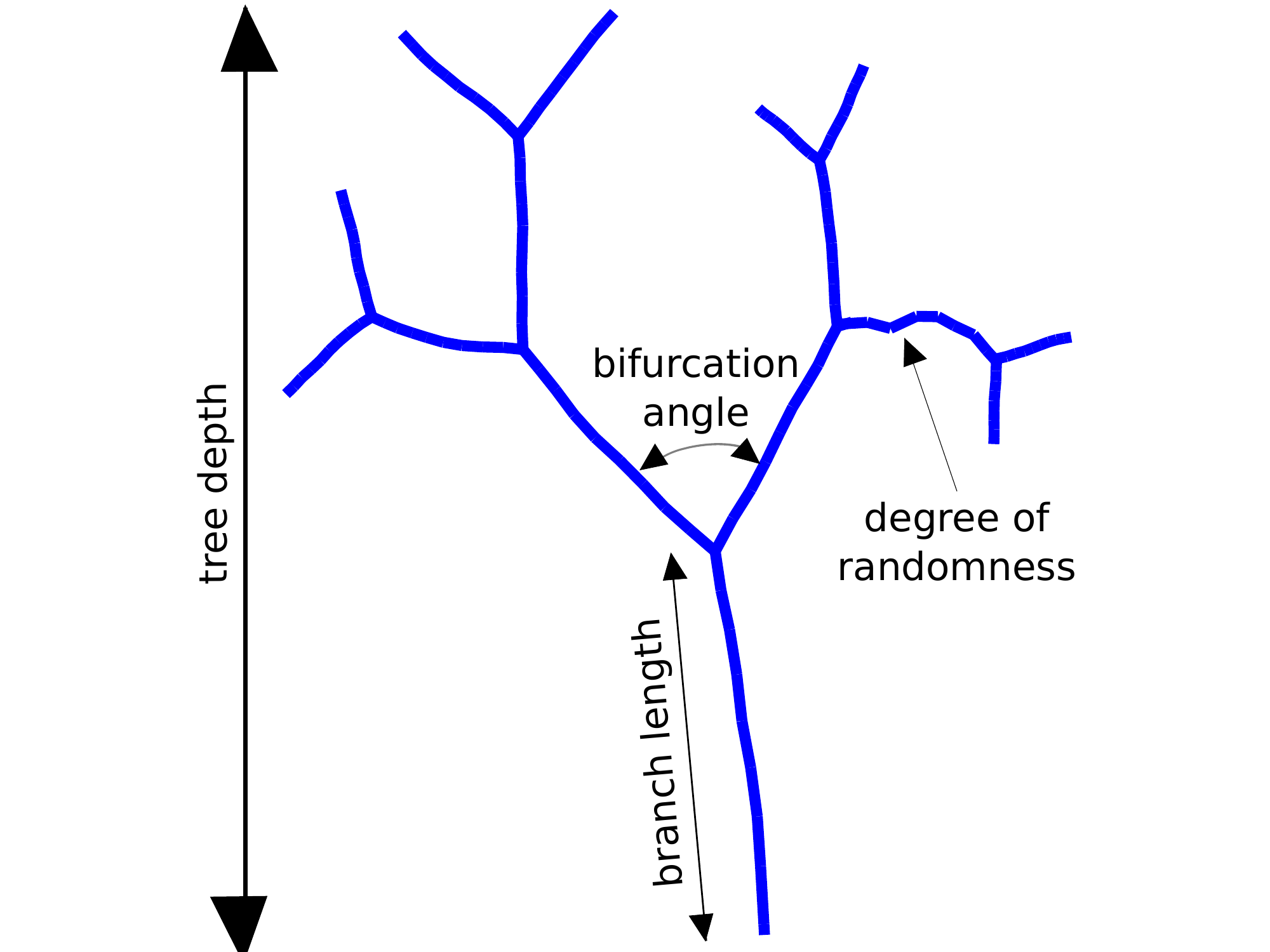}
 \caption{Definition of growth parameters for artificial random trees.}
 \label{fig:randomtree}
\end{figure}

The random trees that were used for testing the algorithm were generated with software developed within the Blue Brain Project (BBP). Each tree consists of branches, i.e., paths between two branch points, which are generated based on a random walk in $\mathbb{R}^3$. The position of the walk at each step is given as a weighted sum of a predefined direction $d_n$ and a simple random walk $X_n$:

$X_{n+1} = w_s \cdot \left((1-D_r) \cdot d_n + D_r \cdot X_n \right),$

\noindent where $w_s$ is the step size, and $D_r$ defines the randomness of each step. For $D_r=0$ the branch is a straight line, while for $D_r=1$ the branch is a simple random walk (SRW). The number of steps is given by the preselected branch length $B_l$. Once the number of steps is reached, the tree bifurcates, i.e., two new branches are created. The angle between the initial points of the branches is defined by the bifurcation angle $B_a$. The tree generated this way is binary, i.e., every leaf has the same depth, since new branches are added at every branch point until the preselected tree depth $T_d$ (i.e., the number of nodes from a leaf to the root node) is reached. The total number of branches in the tree is then $2^{T_d} - 1$. For example, the tree in Fig ~\ref{fig:randomtree} has $T_d=4$ and consists of $2^4 - 1 = 15$ branches.

This set of parameters $\{ T_d, B_l, B_a, D_r \}$ defines the global properties of the tree. Random trees that are generated with the same set of parameters share common morphometric properties, but have unique spatial structures, due to the stochastic component of the growth. This allowed us to check the effectiveness of the algorithm at identifying sets of trees that have been generated with the same input parameters $\{ T_d, B_l, B_a, D_r \}$ and that differ only in the random seed.




\subsection{Random trees results} \label{sec:rtresults}

We defined a control group as a set of trees generated with predefined parameters ($T_d=5$, $B_l=10$, $B_a=\pi/4$, $D_r=10\%$) but independent random seeds. Then, we varied each parameter individually to generate groups of trees that differed from the control group in only one property. For all trees we extracted the persistence barcode using Algorithm 1. A comparison of the distances $d_{Bar}$ between each tree's barcode and the barcodes of the trees in every groups constitutes one trial. The trial is successful if the minimum barcode distance occurs for a tree and the group generated with the same parameters. The overall accuracy of the TMD to separate groups of trees generated with different values for each of the described parameters is given in the following table:

\begin{center}
\begin{tabular}{| c | c | c | c |}
  \hline
  $T_d:$ ($4, 6, 8$)& $B_a:$ ($\frac{\pi}{4}, \frac{\pi}{2}, \pi $) & $B_l:$ ($5, 10, 30$) & $D_r:$ ($0.01, 0.10, 0.90$)\\
  \hline				
  $99 \pm 2 \%$ & $94 \pm 2 \%$ & $99 \pm 1 \%$ & $77 \pm 9 \%$\\
  \hline  
\end{tabular}
\end{center}

\section{Software availability}
\label{sec:soft}

\quad The software used for the extraction of persistence barcodes will be made available upon publication, under the BSD-3 license, more details in \url{ https://opensource.org/licenses/BSD-3-Clause}. The code will be accessible in \url{https://github.com/BlueBrain/Topaz}.

\section{Data provenance} \label{sec:data}
The artificial random trees used in Fig 1 and Fig 3 were generated by software developed in BBP and can be made available upon request.
The biological morphologies used in Fig 1, Fig 2 and Fig 5 are provided from LNMC, EPFL \cite{Romand:2011}. The biological morphologies used in Fig 4 were downloaded from http://neuromorpho.org/ . In particular, cat neurons were provided by \cite{Rose:1995}, dragonfly neurons by \cite{Gonzalez:2013}, drosophila neurons by \cite{Chiang2011}, mouse neurons by \cite{Badea:2011} and rat neurons by \cite{Romand:2011}.

\bibliographystyle{unsrt} 
\bibliography{Topaz_paper} 

\end{document}